\documentclass{llncs}
\usepackage{amsmath}
\usepackage{amssymb}
\usepackage{algorithm}
\usepackage{algpseudocode}
\usepackage{graphicx}
\usepackage{color}
\usepackage{soul}
\usepackage{subfig}

\newcommand{\ket}[1]{|#1\rangle}

\newcommand{\braket}[2]{\langle #1|#2 \rangle}
\newcommand{\half}{\frac{1}{2}}
\newcommand{\sqrttwo}{\frac{1}{\sqrt{2}}}
\newcommand{\sqrtn}[1]{\frac{1}{\sqrt{2^{#1}}}}
\newcommand{\xor}{\oplus}
\newcommand{\A}{\mathcal{A}}

\renewcommand{\to}{\longrightarrow}
\newcommand{\autof}{\breve{f}}

\newcommand{\hatf}{\hat{f}}
\newcommand{\E}{\mathbb{E}}
\newcommand{\Ftwo}[1]{\mathbb{F}_2^{#1}}

\begin{document}

\title{Efficient Quantum Algorithms related to Autocorrelation Spectrum} 
\author{Debajyoti Bera\inst{1}\thanks{Corresponding author} \and Subhamoy Maitra\inst{2} \and Tharrmashastha SAPV\inst{1}}
\institute{Computer Science and Engineering Department, IIIT-D, New Delhi 110020, India. \email{dbera@iiitd.ac.in, tharrmashasthav@iiitd.ac.in} \and Applied Statistics Unit, Indian Statistical Institute, 203 B T Road, Kolkata 700108, India. \email{subho@isical.ac.in}}
\maketitle

\begin{abstract}
In this paper, we propose efficient probabilistic algorithms for several problems regarding the autocorrelation 
spectrum. First, we present a quantum algorithm that samples from the Walsh spectrum of any derivative of $f()$. 
Informally, the autocorrelation coefficient of a Boolean function $f()$ at some point $a$ measures the average 
correlation among the values $f(x)$ and $f(x \oplus a)$. The derivative of a Boolean function is an extension of 
autocorrelation to correlation among multiple values of $f()$. The Walsh spectrum is well-studied primarily due 
to its connection to the quantum circuit for the Deutsch-Jozsa problem. We extend the idea to 
``Higher-order Deutsch-Jozsa" quantum algorithm to obtain points corresponding to large absolute values in 
the Walsh spectrum of a certain derivative of $f()$. Further, we design an algorithm to sample the input points 
according to squares of the autocorrelation coefficients. Finally we provide a different set of algorithms for 
estimating the square of a particular coefficient or cumulative sum of their squares.
\end{abstract}
{\bf Keywords:} Autocorrelation, Boolean function, Cryptology, Quantum computing, Walsh spectrum.

\section{Introduction}\label{sec:intro}
Boolean functions are very important building blocks in cryptology, learning theory and coding theory. 
Different properties of Boolean functions can be well understood by different spectra; specifically, Walsh and autocorrelation spectra are two most important tools for cryptographic purposes. 
For a Boolean function $f()$, these spectra can be thought as the list of all values of the Walsh transform and autocorrelation transform, respectively, of $f()$. 
We use Walsh coefficients and autocorrelation coefficients to indicate the individual values in those spectra.

Shannon related these spectra to confusion and diffusion of cryptosystems long ago~\cite{Shannon:2001:MTC:584091.584093}. Confusion of a Boolean function used in a cryptosystem can be characterized by a Walsh spectrum with low absolute values -- such functions are known to resist linear cryptanalysis\cite{chabaud1994links}; similarly, functions with less diffusion (high absolute value in the autocorrelation spectrum) may make a cryptosystem vulnerable against differential attacks (see for example ~\cite{DBLP:journals/tit/TangM18} and the references therein). 
Walsh spectrum (often referred to as Fourier spectra for Boolean functions) has been shown to be useful for learning Boolean functions as well~\cite{Mansour1994}.

Analyzing these spectra and designing functions with specific spectral properties are therefore important tasks. 
This problem becomes challenging for large functions. 
Such large functions may arise while modelling a complete stream or block cipher as a Boolean function with number of inputs equal to the key size in bits. 
Modelling such a complicated Boolean function by analysing the spectra is clearly elusive~\cite{DBLP:journals/dcc/SarkarMB17}. 
In classical domain, for an $n$-input 1-output Boolean function, generation of complete Walsh or autocorrelation spectrum requires $O(2^n)$ space and $O(n2^n)$ time. 
Needless to mention that for analysing a cipher or learning a Boolean function, it is easier to locate the points if there are high coefficients in a spectrum. 
Thus it makes sense to design techniques for sampling points with high coefficients and estimate the high coefficients in which a Boolean function can be used 
only as a black-box.

The motivation in cipher design is to obtain a Boolean function for which the maximum absolute value in both the spectra is minimized (for autocorrelation we consider non-zero points only). 
While there are many such examples and constructions of such functions in literature related to combinatorics, cryptography and coding theory, such Boolean functions are not implemented in a straightforward manner such as simple circuits or truth/look-up tables. 
This is because it is very hard to implement a complex Boolean function on large number of variables (say 160) in this manner due to exponential circuit size. 
For example, in stream cipher (one may also consider the specific example of Grain v1~\cite{Grain_h_function}), 
LFSR/NFSRs (Linear/Nonlinear Feedback Shift Registers) are used. 
The secret key (say 80 bits) and the public IV (Initialization Vector, say 80 bits again) are loaded in the initial state. Then the initial state is evolved as a Deterministic Finite Automaton for many (say 160 or 200) steps. The output bit is generated by combining some selected bits (say 15) from the LFSR/NFSRs. Then we start generating the output bits which is used as key stream bits for cryptographic purposes. 
Now if you consider the initial key and IV as the inputs to a Boolean function and the key stream bit at any instance as an output, this is a Boolean function with 160 input bits and one output bit. 
Modelling such a complicated Boolean function by analysing the spectra is practically not possible. 
For more details, one may refer to~\cite{DBLP:journals/dcc/SarkarMB17}. However, if the complete circuit can be implemented in quantum paradigm, then one may have much better efficiency in mounting the attacks.

The situation is well settled for the Walsh spectrum.
Walsh spectrum of a function $f:\{0,1\}^n \to \{0,1\}$ is defined as the following function~\footnote{\label{foot:1}The normalization factor used depends upon the application but has no bearing on properties of interest.} from $\{0,1\}^n$ to $\mathbb{R}[-1, 1]$ in which $x \cdot y$ stands for the $0-1$ valued expression 
\smallskip
$\xor_{i=1 \ldots n} x_i y_i$:\\
\centerline{$\displaystyle\mbox{for $y \in \{0,1\}^n$,\quad} \hatf(y) = \frac{1}{2^n} \sum_{x \in \{0,1\}^n} (-1)^{f(x)} (-1)^{y \cdot x}$}
\smallskip

\begin{figure}
    \centering
    \includegraphics[width=0.4\linewidth]{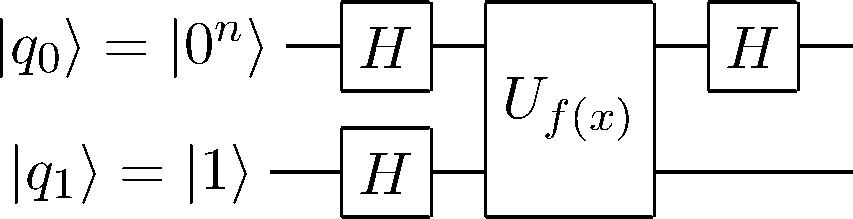}
    \caption{Circuit for Deutsch-Jozsa algorithm (without measurement)\label{fig:DJ}}
\end{figure}

The Deutsch-Jozsa algorithm~\cite{Deutsch1992}, even though usually described as solving a different problem, makes only {\em one} query to $U_f$ (a standard unitary implementation of $f()$) and at the end, puts the second register in the state $\ket{1}$ and the 
first register in the state $\sum_{z \in \Ftwo{n}} \hatf(z) \ket{z}$; the quantum circuit for the same is illustrated in Figure~\ref{fig:DJ}.
Measuring the second register in the standard basis generates a state $\ket{z}$ with probability $\hatf(z)^2$. Note that Walsh coefficients do satisfy $\sum_z \hatf(z)^2 = 1$ (this is due to Parseval's theorem); thus the Deutsch-Jozsa algorithm can be considered as an efficient sampling algorithm for Walsh coefficients~\cite{maitra2005deutsch}.
So if one can implement a stream cipher (a Boolean function) as a quantum oracle~\cite{QuantumAES2016}, then it is possible to sample high points in a Walsh spectrum in constant time with linear number of gates and that enables us to answer several questions related to the spectrum~\cite{Xie2018}.

In contrast to the Walsh spectrum, the autocorrelation spectrum is less studied.
It is defined as the following transformation~\footnotemark[\value{footnote}] from $\{0,1\}^n$ to $\mathbb{R}[-1, 1]$.\\
\centerline{$\displaystyle \mbox{for $a \in \{0,1\}^n$, \quad} \autof(a) = \frac{1}{2^n} \sum_{x \in \{0,1\}^n} (-1)^{f(x)} (-1)^{f(x \xor a)}$}
\smallskip

The entire autocorrelation spectrum can be obtained by first computing the Walsh spectrum (using the well-known ``fast Walsh-Hadamard transform'' algorithm), then squaring each of the coefficients, and finally applying the same transform once more on this squared spectrum. 
This approach requires $2^n$ many calls to $f()$, $n2^n$ other operations and space complexity of $2^n$. 

However, a question remains that {\em what} can be found out about the autocorrelation spectrum in $o(2^n)$, preferably polynomial, time.

\begin{itemize}
    \item Especially, can we identify the points with high coefficients? 
    \item Can we estimate a particular coefficient? 
\end{itemize}

Counting and sampling often go hand-in-hand, so one would also like to sample from a distribution proportional to the coefficients.
It should be noted that $\sum_a \autof(a)^2 \in [1,2^n]$ unlike Walsh coefficients, therefore, it appears difficult to get a quantum sampling algorithm like Deutsch-Jozsa as an immediate corollary.

The quantum algorithms we propose in this paper address these questions. Naturally, in terms of autocorrelation spectrum, such algorithms will be able to expose the weaknesses of a Boolean function (used in a cryptographic primitive) better than the classical 
approaches. There are quite a few important research results related to quantum cryptanalysis of symmetric ciphers~\cite{DBLP:conf/crypto/KaplanLLN16,DBLP:journals/tosc/KaplanLLN16,DBLP:conf/asiacrypt/ChaillouxNS17}. 
A recent work~\cite{DBLP:conf/asiacrypt/Leander017} in this direction considered merging the ideas from Grover's~\cite{Grover1998} and Simon's~\cite{Brassard1997} algorithms. 
However, there has been no specific attempt to solve concrete problems related to the autocorrelation spectrum.
This we present in this paper.

One of the ideas used by us is that of amplitude amplification which is the underlying engine behind Grover's algorithms. 
However, our approach is very different from that of Simon's algorithm even though it is tempting to use this algorithm since $\autof(a)=1$ iff $f(x) = f(x \xor a)$ for all $x$ and the latter is one of main promises held by $f$ in the Simon's problem. 
First, another condition on $f$, i.e., if $f(x) = f(y)$ then $x=y \xor a$, may not necessarily hold for 1-bit functions and secondly, Simon's algorithm is specifically designed for finding any such $a$ and not sampling according to a distribution proportional to $\autof(a)$.

Another important measure related to autocorrelation spectrum is the sum-of-squares indicator. Naturally it is better if this value 
is low. 
\begin{definition}[Sum-of-squares indicator]\label{defn:s_f}
The sum-of-squares indicator for the characteristic of $f$ is defined as
$$\sigma_f = \sum_{a \in \Ftwo{n}} \autof(a)^2$$
\end{definition}
It is known that $1 \le \sigma_f \le 2^n$. In particular, $\sigma_f=1$ if $f$ is a Bent function and $\sigma_f=2^n$ if $f$ is a linear function. A small $\sigma_f$ indicates that a function satisfies the {\em global avalanche criteria} (GAC).

\subsection{Outline}
The results in this paper answer the questions of sampling and estimation that
were raised above.

In Section~\ref{sec:derv-sampling} we present a generalization of the Deutsch-Jozsa problem that we name as ``Higher-order Deutsch-Jozsa''(HoDJ), which is related to the derivatives of a Boolean function. Higher-order derivatives capture the correlation among multiple output values of the same function and is important for constructing cryptographic hash functions that are resistant to linear attack, differential attack, cube attack, etc.
% \marginpar{@subho:add more motivation and citations for higher-order}

We then discuss a quantum algorithm whose output is a random sample from a distribution that is proportional to the Walsh coefficients of 
any specific higher-order derivative. For $k$-th order derivative, the algorithm uses only $n+1$ {\em additional} qubits, makes $2^k$ calls to the function and uses altogether $O(n2^k)$ gates that is a meagre fraction compared to the usual exponential (in $n$) time and space complexity seen in classical algorithms.

The first-order derivative is also known as the autocorrelation spectrum so this sampling algorithm can be used to generate samples according to the distribution of the {\em Walsh coefficients} of the autocorrelation coefficients. By making a subtle observation, we show how to actually sample according to the autocorrelation spectrum itself. We are not aware of any classical sampling algorithm for the autocorrelation spectrum and the only algorithm known for generating the entire spectrum, which involves computing Walsh transformation twice and is no doubt an overkill for the task of sampling, incurs $\Theta(2^n)$ space complexity and $\Theta(n2^n)$ time complexity. In comparison to it, our quantum algorithm has $O\left( n\tfrac{2^{n/2}}{\sqrt{\sigma_f}} \log \tfrac{1}{\delta} \right)$ time complexity (exhibiting a quadratic speedup) and $2n+1$ space complexity; here $\delta$ indicates the probability of failure. If $\sigma_f$ is not too small, say $\tfrac{2^n}{poly(n)}$, then the time complexity shows an exponential speedup over the classical one. We explain this algorithm for autocorrelation sampling and discuss its properties in Section~\ref{sec:sampling}.

We next move on to estimating algorithms in Section~\ref{sec:estimation}. First, in Subsection~\ref{subsec:autocor-estimation} we give a quantum algorithm to estimate the autocorrelation coefficient at any given point with high accuracy, denoted $\epsilon$, and low error, denoted $\delta$. Our algorithm makes $\Theta\left( \tfrac{1}{\epsilon} \log \tfrac{1}{\delta} \right)$ calls to the function (rather, a quantum oracle for the same). This is almost square-root of the known classical complexity of $O\left( \tfrac{1}{\epsilon^2} \log \tfrac{1}{\delta} \right)$. We explain why the sampling techniques that we designed {\em cannot} be used to design an efficient estimation algorithm, and instead, design our algorithm using the idea of a ``swap-test''.

Our final contribution is a quantum algorithm to estimate the sum-of-squares $\sigma_f$; this we describe in Subsection~\ref{subsec:est_s_f}. We explain that a classical sampling based approach requires $O\left( \tfrac{2^{2n}}{\epsilon^2} \log \tfrac{1}{\delta} \right)$ calls to $f$ ($\epsilon$ would generally be greater than $1$ for estimating $\sigma_f$ since $\sigma_f \in [1,2^n]$) and then describe a quantum approach that displays quadratic speedup and only makes $O\left( \tfrac{2^{n}}{\epsilon} \log \tfrac{1}{\delta} \right)$ calls.

%{\color{red} @subho: Can we connect these concrete results to designing/attacking of cryptographic hash functions ... just by leaving some pointers or throwing some keywords?}

\section{Sampling from Higher-order Derivative}\label{sec:derv-sampling}
Higher-order derivatives of a Boolean function was explicitly introduced, in the context of cryptanalysis, by Lai~\cite{Lai1994}.

\begin{definition}[Derivative]
Given a point $a \in \{0,1\}^n$, the (first-order) derivative of an $n$-bit function $f$ {\em at $a$} is defined as $$\Delta f_a(x) = f(x \oplus a) \xor f(x)$$ For a list of points $\A=(a_1, a_2, \ldots, a_k)$ (where $k \le n$) the $k$-th derivative of $f$ at $(a_1,a_2, \ldots, a_k)$ is recursively defined as $$\Delta f^{(k)}_{\A}(x)  = \Delta f_{a_k} (\Delta f_{a_1, a_2, \ldots, a_{k-1}}^{(k-1)}(x)),$$ where $\Delta f_{a_1, a_2, \ldots, a_{k-1}}^{(k-1)}(x)$ is the $(k-1)$-th derivative of $f$ at points $(a_1,a_2, \ldots, a_{k-1})$. 
The $0$-th derivative of $f$ is defined to be $f$ itself.
\end{definition}
Higher-order derivatives form the basis of many cryptographic attacks, especially those that generalize the differential attack technique against block ciphers such as Integral attack, AIDA, cube attack, zero-sum distinguisher, etc.
These attacks mostly revolve around the algebraic degree of a higher-order derivative. Let $deg(f)$ denote the algebraic degree of some function $f$. 
It is known that $deg(\Delta f^{(i+1)}) \le deg(\Delta f^{(i)}) - 1$ and if $f$ is an $n$-bit function then $\Delta f^{(n)}$ is a constant function. 
Thus if a function has the degree of its $i$-th order derivative, at some $(a_1, a_2, \ldots a_i)$, to be a constant, then this fact is essentially a beacon for mounting an attack if $i \ll n$.
Therefore, it is central to study the algebraic degree and other properties of higher-order derivatives, and to the best of our knowledge, we provide the first algorithms for these tasks.

Specifically, we show how to efficiently sample from the Walsh-Hadamard spectrum of the $i$-th order derivative. 
This allows us to estimate if a higher-order derivative of $f$ is biased towards any linear function, thereby partly answering the question above since the Walsh-Hadamard transform of a linear function is constant.

Despite the complicated expression for computing $\Delta f^{(k)}$, it has an equivalent expression that we shall use for our results.
For any multiset $S$ of points (including $S = \emptyset$), define the notations $X_s = \bigoplus_{a \in S} a$ and $f(x \xor S) = f(x \xor X_s)$. 
In the case of $S = \emptyset$, it can be noted that $X_s$ is the empty string and hence $f(x \xor S) = f(x)$.
The $i$-th derivative of $f$ at $\A=(a_1,a_2, \ldots a_i)$ can be shown\footnote{The proof is present in~\cite{Lai1994}} to be $$\Delta f_{\A}^{(i)} (x) = \bigoplus_{S \subseteq A} f(x \xor S)$$ where $S \subseteq A$ indicates all possible sub-lists of $\A$ (including duplicates, if any, in $\A$).
For example, the second-order derivative at a pair of points $(a,b)$ can be written as $$\Delta f^{(2)}_{(a,b)} = {f(x)} \xor {f(x \xor a)} \xor {f(x \xor b)} \xor {f(x \xor a \xor b)}.$$
For the sake of brevity, we will drop the superscript $(i)$ if it is clear from the list $\A$.

Now we describe a quantum circuit that generates the Walsh-Hadamard spectrum of the $k$-derivative of an $n$-bit function $f$ at some set of points $\A = (a_1, a_2, \ldots a_k)$.
We refer to the circuit as $HoDJ^k_n$ (``Higher-order Deutsch-Jozsa'').

For calling $f$ we use the standard unitary operator $U_f : \ket{x}\ket{b} \mapsto \ket{x} \ket{b \xor f(x)}$ where $x \in \{0,1\}^n$ and $b \in \{0,1\}$. We use $\ket{+}$ and $\ket{-}$ to denote the states $\sqrttwo(\ket{0}+\ket{1})$ and $\sqrttwo(\ket{0}-\ket{1})$, respectively;
observe that $U_f \ket{x} \ket{+} = \ket{x} \ket{+}$ and $U_f \ket{x} \ket{-} = (-1)^{f(x)} \ket{x} \ket{-}$.

The circuit for $HoDJ^k_n$ acts on $k+2$ registers, $R_1, \ldots R_k, R_{k+1}, R_{k+2}$ that are initialized as 
\begin{itemize}
    \item $R_{1}$ has one qubit that is initialized to $\ket{1}$,
    \item $R_2$ consists of $n$-qubits that is initialized to $\ket{0^n}$,
    \item and each of $R_3 \ldots R_{k+2}$ consists of $n$-qubits in which $R_{2+t}$ is initialized to $a_t$ of $\A$.
\end{itemize}

The circuit itself is a generalization of the quantum circuit for the Deutsch-Jozsa problem~\cite{Deutsch1992} and uses the ability of this circuit to generate a distribution of Walsh-Hadamard coefficients that was explained earlier.

\begin{figure}[!ht]
    \centering\leavevmode
    \includegraphics[width=0.75\linewidth]{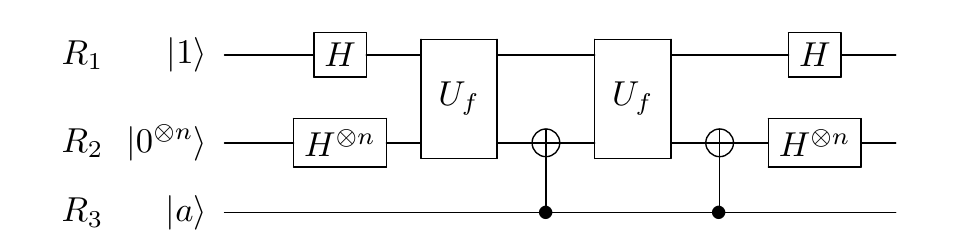}
    \caption{Circuit for $1^{st}$-order Walsh-Hadamard derivative sampling\label{fig:auto-sampling}} 
\end{figure}

Figure~\ref{fig:auto-sampling} shows the quantum circuit for $HoDJ^1_n$; for this problem, $\A$ is a singleton set, say $\{a\}$.
The evolution of the quantum state as the operators are applied is as follows:
\begin{align*}
    \mbox{Initial State} :  ~& \ket{1}\ket{0^{ n}} \ket{a}\\
    \xrightarrow{H \otimes H^{n}} ~& \frac{1}{\sqrt{2^n}}\sum_x\ket{-}\ket{x}\ket{a}\\
    \xrightarrow{U_f} ~& \frac{1}{\sqrt{2^n}}\sum_x{(-1)^{f(x)}\ket{-}\ket{x}}\ket{a} \\
    \xrightarrow{CNOT_2^3} ~& \frac{1}{\sqrt{2^n}}\sum_x{(-1)^{f(x)}}\ket{-}\ket{x\oplus a}\ket{a}\\
    \xrightarrow{U_f} ~& \frac{1}{\sqrt{2^n}}\sum_x{(-1)^{f(x)\oplus f(x\oplus a)}}\ket{-}\ket{x\oplus a}\ket{a}\\
    \xrightarrow{CNOT_2^3} ~& \frac{1}{\sqrt{2^n}}\ket{-}\sum_x{(-1)^{f(x)\oplus f(x\oplus a)}}\ket{x}\ket{a}\\
    \xrightarrow{H\otimes H^{n}} ~& \ket{1}\sum_y\Big[\frac{1}{2^n}\sum_x(-1)^{(x\cdot y)}(-1)^{f(x)\oplus f(x\oplus a)}\Big]\ket{y}\ket{a}\\
    = ~& \ket{1} \sum_y \widehat{\Delta f_a}(y) \ket{y} \ket{a}
\end{align*}

Therefore, at the end of the circuit $R_2$ can be found to be in a state
$\ket{y}$ with probability $\widehat{\Delta f_a}(y)^2$ thus accomplishing the
objective of sampling according to the Walsh-Hadamard distribution of the
1st-order derivative of $f$.

Next, an illustration of $HoDJ^2_n$ corresponding to the 2nd-order derivative is presented in Figure~\ref{fig:sampling-derivative} in which we use $\A=(a,b)$.
We show the state of this circuit after each layer of operators.

\begin{figure}[!ht]
    \centering\leavevmode
    \includegraphics[width=0.9\linewidth]{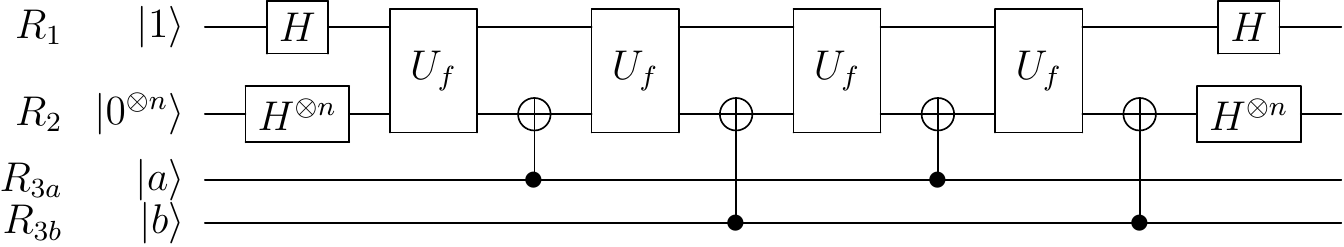}
    \caption{Circuit for Walsh-Hadamard sampling of
    $2^{nd}$-order derivative\label{fig:sampling-derivative}} 
\end{figure}

\begin{align*}
    \mbox{Initial State} :~& \ket{1}\ket{0^n}\ket{a}\ket{b}\\
    \xrightarrow{H \otimes H^n}~ & \sqrtn{n} \sum_{x \in \{0,1\}^n} \ket{-}\ket{x}\ket{a,b}\\
    \stackrel{U_f}{\to}~ & \sqrtn{n} \sum_x (-1)^{f(x)} \ket{-}\ket{x,a,b}\\
    \xrightarrow{CNOT^3_2}~ & \sqrtn{n} \sum_x (-1)^{f(x)} \ket{-} \ket{x\xor a}\ket{a,b}\\
    \xrightarrow{U_f}~ & \sqrtn{n} \sum_x (-1)^{f(x) \xor f(x\xor a)} \ket{-} \ket{x\xor a}\ket{a,b} \\
    \xrightarrow[U_f]{CNOT^4_2,}~ & \sqrtn{n} \sum_x (-1)^{f(x) \xor f(x \xor a) \xor f(x\xor a \xor b)} \ket{-} \ket{x \xor a \xor b} \ket{a,b}\\
    \xrightarrow[U_f]{CNOT^3_2}~ & \sqrtn{n} \sum_x (-1)^{\bigoplus_{S \subseteq \{a,b\}} f(x \xor S)}
    \ket{-} \ket{x \xor b} \ket{a,b} \\
    \xrightarrow{CNOT^4_2}~ & \sqrtn{n} \sum_x (-1)^{\bigoplus_{S \subseteq \{a,b\}} {f(x \xor S)}} \ket{-} \ket{x,a,b} \\
    \xrightarrow{H \otimes H^n}~ & \ket{1} \sum_y \Big[ \frac{1}{2^n} \sum_x (-1)^{x\cdot y} (-1)^{\bigoplus_{S \subseteq \{a,b\}}
	{f(x \xor S)}} \Big] \ket{y} \ket{a,b}
\end{align*}

Measuring $R_2$ at the end will collapse it into $\ket{y}$ for some $y \in \{0,1\}^n$ with probability $\Pr[y]=\Big[ \frac{1}{2^{n}} \sum_x (-1)^{x\cdot y} \Delta f_{(a,b)}(x) \Big]^2$ $=\widehat{\Delta f_{(a,b)}}(y)^2$ that is the square of the Walsh coefficient of $\Delta f_{(a,b)}$ (2nd-order derivative function) at the point $y$.

The circuit can be generalized to higher values of $k$ in a straight forward manner.
The following theorem formalizes this result where we ignore the first register since that contains an ancillary qubit which is reset to its initial state at the end of the computation. 
For counting the number of gates, please note that each of the CNOT gates shown in Figure~\ref{fig:sampling-derivative} actually consists of $n$ 2-qubit CNOT gates applied in parallel.

\begin{theorem}\label{thm:hodj} 
For any $\A = (a_1, a_2, \ldots a_k)$ such that $a_i \in \{0,1\}^n$ $\forall i$, the $HoDJ^k_n$ circuit uses $n+1$ initialized ancilla qubits, employs $k$ registers corresponding to the points in $\A$, makes $2^k$ calls to $U_f$, $\Theta(n2^k)$ calls to $H$ and $CNOT$ gates, has a depth of $2(2^k+1)$ and operates as follows  $$\ket{0^n}\ket{a_1}\ldots\ket{a_k} \xrightarrow{HoDJ^k_n} \sum_y  \widehat{\Delta f_{\A}}(y) \ket{y} \ket{a_1}\ldots\ket{a_k}$$
\end{theorem}

\begin{proof}
    The circuit is a generalization of those illustrated in Figures \ref{fig:auto-sampling} and \ref{fig:sampling-derivative}. At the core is a sub-circuit that we denote by $C'$ and which acts as $$\ket{x}\ket{a_1}\ldots\ket{a_k}\ket{b} \xrightarrow{C'} \ket{x}\ket{a_1}\ldots\ket{a_k} \ket{b \oplus \Delta f_{\A}}(x) = \ket{x}\ket{a_1}\ldots\ket{a_k} \ket{b \oplus 
    \bigoplus_{S \subseteq \A} f(x \oplus S)}$$
Construction of $C'$ uses a {\em binary reflected Gray code} (BRGC, or ``Gray
    code'' in short) for the set of integers $\{0, 1, \ldots, 2^k-1\}$. Such a BRGC
    will be a sequence of $k$-bit strings (codes) $(g_1, g_2, \ldots, g_{2^k})$ such
    that each $g_i$ is unique and every adjacent code differ at exactly one
    position. Integer $0$ is encoded by the code $0^n$ and without loss of
    generality, let $g_{2^k} = 0^n$. Due to the cyclic property of BRGC,
    $g_{1}$ must be some $k$-bit string with Hamming weight 1.

    $C'$ operates in $2^k$ stages. We will use $\ket{\A}$ as a shorthand for
    $\ket{a_1} \ldots \ket{a_k}$. The initial state of the qubits, before stage
    1, is $\ket{x} \ket{\A} \ket{b}$. Observe that $\bigoplus_{S \subseteq \A}
    f(x \oplus S) = \bigoplus_{i=1}^{2^k} f(x \oplus (g_j \cdot \A))$ in which
    we used the notation  $g_j \cdot \A = (g_j)_1 a_1 \oplus (g_j)_2 a_2 \oplus \ldots (g_j)_k
    a_k$ to denote a a linear combination of some of the $a_i$'s.

    The $j$-th stage of $C'$ creates the state $\ket{x \oplus (g_{j} \cdot \A)} ~\ket{\A}~ \ket{b
	\oplus \bigoplus_{i=1}^{j} f(g_{i} \cdot \A)}$ by making the following transformations.
    \begin{align*}
	& \ket{x \oplus (g_{j-1} \cdot \A)} ~\ket{\A}~ \ket{b \oplus
	\bigoplus_{i=1}^{j-1} f(g_{i} \cdot \A)}\\
        \xrightarrow{CNOT} & \ket{x \oplus (g_{j} \cdot \A)} ~\ket{\A}~ \ket{b
	\oplus \bigoplus_{i=1}^{j-1} f(g_{i} \cdot \A)}\\
        \xrightarrow{U_f} & \ket{x \oplus (g_{j} \cdot \A)} ~\ket{\A}~ \ket{b
	\oplus \bigoplus_{i=1}^{j-1} f(g_{i} \cdot \A) \oplus f(g_j \cdot \A)}
    \end{align*}
    The $CNOT$ operation above is justified since $g_{j-1} \cdot \A$ and $g_j
    \cdot \A$ are both linear combinations of some of the $a_i$'s differing by
    exactly one $a_t$. The $CNOT$ uses the corresponding register $\ket{a_t}$ as
    the control register and the first register qubit as the target register.
    This also holds true for stage 1 since $g_1$ has Hamming weight 1.
    Lastly, observe that the final state after the $2^k$-th stage matches the
    one specified above: $\ket{x}
    ~\ket{\A} ~\ket{b \oplus \bigoplus_{S \subseteq \A} f(x \oplus S)}$.

    It is not hard to calculate that $C'$ also makes the following
    transformation if $\ket{b}$ is replaced by $\ket{-}$.
    $$\ket{x} \ket{a_1 \ldots a_k} \ket{-} \xrightarrow{C'}
    (-1)^{\oplus_{S \subseteq \A} f(x \oplus S)} \ket{x}
    \ket{a_1 \ldots a_k} \ket{-} = (-1)^{\Delta f_{\A}(x)} \ket{x}
    \ket{a_1 \ldots a_k} \ket{-}$$

    The circuit for $HoDJ^k_n$ is constructed as
    \begin{align*}
	~& ~\ket{-} \ket{0^n} \ket{a_1 \ldots a_k} \\
	 \xrightarrow{H^n} ~& ~\sqrtn{n}\sum_x \ket{-} \ket{x} \ket{a_1 \ldots a_k} \\
	 \xrightarrow{C'} ~& ~\sqrtn{n}\sum_x \ket{-} (-1)^{\Delta f_{\A}(x)} \ket{x}
	    \ket{a_1 \ldots a_k} = ~\ket{-} \sqrtn{n}\sum_x (-1)^{\Delta f_{\A}(x)} \ket{x}
	    \ket{a_1 \ldots a_k}\\
	 \xrightarrow{H^n} ~& ~\ket{-} \sum_y \widehat{\Delta f_{\A}}(y) \ket{y} \ket{a_1
	\ldots a_k}
    \end{align*}
    For computing the resource usage of $HoDJ^k_n$, observe that $C'$ is
    implemented above using a depth $2\cdot 2^k$ circuit and each of its stages
    employ one $U_f$ gate and $n$ $CNOT$ gates (that act in parallel on all the
    $n$ qubits of the first register and is
    shown as a single $CNOT$ operation above). This completes the proof of the
    theorem.    
\qed
\end{proof}

A quick observation is that $HoDJ^0_n$ essentially generates $\sum_y \hat{f}(y) \ket{y}$ that is exactly the same output as that of the Deutsch-Jozsa circuit and in fact, the circuit for $HoDJ^0_n$ is exactly same as that of the Deutsch-Jozsa circuit for $n$-bit functions.

\section{Autocorrelation Sampling}\label{sec:sampling}

In section~\ref{sec:derv-sampling} we explained how to sample from the higher
order derivatives of a Boolean function. In this section we present an algorithm
to sample according to a distribution that
is proportional to the autocorrelation coefficients of a function; specifically,
we would like to output $\ket{a}$ with probability proportional to
$\autof(a)^2$. We will use the technique presented in
Section~\ref{sec:derv-sampling} for doing so and will use a key observation
stated in this lemma.

\begin{lemma}\label{lemma:auto-sampling}
$\autof(a) = \widehat{\Delta f^{(1)}_{a}}(0^n)$
\end{lemma}

\begin{proof}
    LHS is equal to $\frac{1}{2^n} \sum_x (-1)^{f(x)} (-1)^{f(x\xor a)} =
    \frac{1}{2^n} \sum_x \Delta f^{(1)}_{a}(x)$. Now observe that 
    $\widehat{\Delta f^{(1)}_{a}}(0^n) = \frac{1}{2^n} \sum_x \Delta f^{(1)}_{a}(x)$ and this proves the lemma.
\qed
\end{proof}

The circuit used in Algorithm~\ref{algo:autocor-sampling} is illustrated in Figure~\ref{fig:autocorr-sampling}.

\begin{algorithm}
    \caption{Algorithm for autocorrelation sampling\label{algo:autocor-sampling}}
\begin{algorithmic}[1]
    \State Start with three registers initialized as $\ket{1}$, $\ket{0^n}$,
    and $\ket{0^n}$.
    \State Apply $H^{n}$ to $R_{3}$ to generate the state $\sqrtn{n}\sum_{b
    \in \Ftwo{n}} \ket{1} \ket{0^{n}} \ket{b}$.
    \State Apply $HoDJ^1_n$ on the registers $R_1$, $R_2$ and $R_3$   
    to generate the state
    \Statex $\displaystyle\ket{\Phi} = \sqrtn{n} \ket{1} \sum_{b \in \Ftwo{n}} \sum_{y \in
    \Ftwo{n}} \widehat{\Delta f^{(1)}_{b}}(y) \ket{y} \ket{b}$.
    \State Apply fixed-point amplitude amplification~\cite{YoderFPSearch} on $\ket{\Phi}$ to amplify the probability of observing $R_2$ in the state $\ket{0}$ to $1-\delta$ for any given constant $\delta$
    \State Measure $R_3$ in the standard basis and return the observed outcome
\end{algorithmic}
\end{algorithm}

\begin{figure}[!ht]
	\centering\leavevmode
	\includegraphics[width=.8\linewidth]{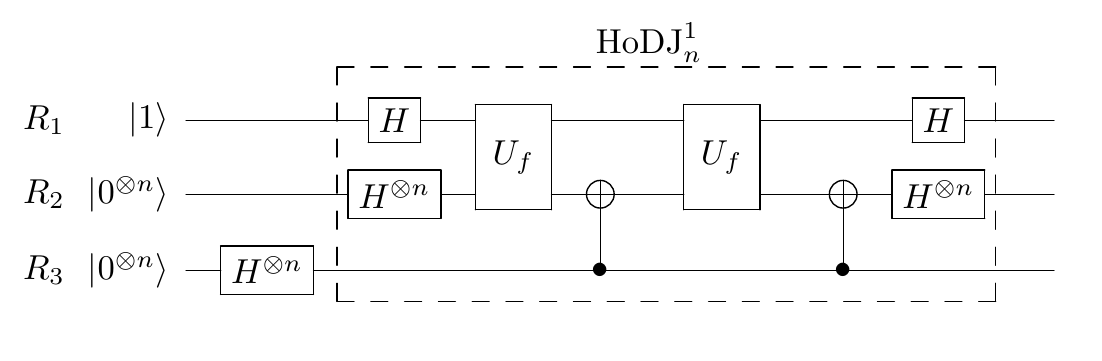}
	\caption{Circuit for autocorrelation sampling\label{fig:autocorr-sampling}}
\end{figure}

\begin{theorem}\label{thm:auto-sampling}
    The observed outcome returned  by Algorithm~\ref{algo:autocor-sampling} is a random sample from the distribution $\{ \autof(a)^2/\sigma_f \}_{a \in \Ftwo{n}}$  with probability at least $1-\delta$. The algorithm makes $O(\frac{2^{n/2}}{\sqrt{\sigma_f}} \log \frac{2}{\delta})$ queries to $U_f$ and uses $O(n\frac{2^{n/2}}{\sqrt{\sigma_f}} \log \frac{2}{\delta})$ gates altogether.
\end{theorem}

\begin{proof}
    We can write the final state of the circuit in Figure~\ref{fig:autocorr-sampling} as 
    \begin{align*}
    \ket{\Phi} = & \sqrtn{n} \ket{1} \sum_{b \in \Ftwo{n}} \sum_{y \in
    \Ftwo{n}} \widehat{\Delta f^{(1)}_{b}}(y) \ket{y} \ket{b} \\
    = & \ket{1} \otimes \ket{0^n} \otimes \left( \sqrtn{n} \sum_b \widehat{\Delta f_b}(0^n) \ket{b}\right) + \sum_y \ket{1} \ket{y} \otimes \left( \sqrtn{n} \sum_b \widehat{\Delta f_b}(y) \ket{b}\right)\\
	= & \ket{1} \otimes \ket{0^n} \otimes \left( \sqrtn{n} \sum_b \autof(b) \ket{b}\right) + \sum_y \ket{1} \ket{y} \otimes \left( \sqrtn{n} \sum_b \widehat{\Delta f_b}(y) \ket{b}\right)\\
    \end{align*}
     
    Suppose we denote the normalized state $\frac{1}{\sqrt{\sigma_f}} \sum_b \autof(b) \ket{b}$ by $\ket{\Phi'}$ and the state $\sqrtn{n} \sum_b \widehat{\Delta f_b}(y) \ket{b}$ by $\ket{\Phi''_y}$. Then, using Lemma \ref{lemma:auto-sampling} we can rephrase $\ket{\Phi}$ as
    $$ \ket{\Phi} = \sqrt{\frac{\sigma_f}{2^n}} \ket{1} \otimes \ket{0^n} \otimes \ket{\Phi'} + \sum_y \ket{1} \ket{y} \ket{\Phi'_y}$$
    and the probability of observing $R_2$ in state $\ket{0^n}$ as $\sigma_f/2^n$.
    
    Fixed-point amplitude amplification will make $O(\frac{2^{n/2}}{\sqrt{\sigma_f}} \log \frac{2}{\delta})$ calls to the circuit in Figure~\ref{fig:autocorr-sampling} and ensure that the amplitude of the state $\ket{1}\ket{0^n}\ket{\Phi'}$ is at least $\sqrt{1-\delta}$. Therefore, after amplification $R_3$ will be in the state $\ket{\Phi'}$ with probability at least $1-\delta$, and when that happens, the observed state upon measuring $R_3$ would be some $\ket{b}$ with probability $\autof(b)^2/\sigma_f$ --- that is, a sample from the autocorrelation distribution.

    The number of queries required for the whole process is the number of times that amplitude amplification calls the circuit ($O(\frac{2^{n/2}}{\sqrt{\sigma_f}} \log \frac{2}{\delta})$) multiplied by the number of calls to $U_f$ made by the circuit (which is only two). The total number of gates involved is also obtained in a similar manner along with the observation that the circuit uses $\Theta(n)$ which is evident from Figure~\ref{fig:autocorr-sampling}. \qed
\end{proof}

\section{Estimation Algorithms}\label{sec:estimation}

The main problem here is to estimate, with high accuracy and small error (if
any), important functions of an autocorrelation spectrum.

For these algorithms we use the quantum technique of amplitude estimation. We
use a particular version that was recently presented for estimating the
probability of ``success'' of a quantum circuit (where success corresponds to
the output state of the circuit to be in a certain subspace) with additive
accuracy.

\begin{lemma}[\cite{bera2019error}\label{lemma:ampest}]
    Let $\A$ be a quantum circuit without any measurement and let $p$ denote the
    probability of observing its output state in a particular subspace. There is a quantum
    algorithm that makes a total of $\Theta(\frac{\pi}{\epsilon}\log\frac{1}{\delta})$
    calls to (controlled)-$\A$ and returns an estimate $\tilde{p}$ such that,
    $$\Pr[\tilde{p} - \epsilon \le p \le \tilde{p} + \epsilon] \ge 1-\delta$$
    for any accuracy $\epsilon \le \tfrac{1}{4}$ and error $\delta < 1$.
\end{lemma}

\subsection{Autocorrelation Estimation}\label{subsec:autocor-estimation}
The objective of this section is to estimate the value of $|\autof(a)|$ for any particular $a \in \{0,1\}^n$; this is identical to
estimating $|\autof(a)|^2$.

First, observe that $\autof(a) = \frac{1}{2^n} \sum_x (-1)^{f(x)} (-1)^{f(x \xor
a)} = \E_x[X_x]$ where the $\pm1$-valued random variable $X_x = (-1)^{f(x) \xor
f(x \xor a)}$ is defined for $x$ chosen uniformly at random from $\{0,1\}^n$.
Therefore, the number of samples needed if we were to classically estimate
$\autof(a)$ with accuracy $\epsilon$ and error $\delta$ is
$O(\frac{1}{\epsilon^2}\log\frac{1}{\delta})$.

The quantum circuit in Figure~\ref{fig:autocorr-sampling} can also be used to estimate $|\autof(a)|$, rather, $\autof(a)^2/2^n$. Recall that the probability of observing $R_2$ in the state $\ket{0^n}$ and $R_3$ in the state $\ket{a}$ (without any amplification) is $\frac{\autof(a)^2}{2^n}$ (refer to the proof of Theorem~\ref{thm:auto-sampling}). Let $F$ denote $\frac{\autof(a)^2}{2^n}$, $\epsilon$ denote the desired accuracy and $\delta$ denote the desired probability of error. Call the algorithm in Lemma~\ref{lemma:ampest} to obtain an estimate $F'$ of $F$ with an accuracy $\epsilon'$ and error probability $\delta$. 
We know from the lemma that with high probability $F' - \epsilon' \le F \le F' + \epsilon'$ which implies that $2^n F' - 2^n \epsilon' \le \autof(a)^2 \le 2^n F' + 2^n \epsilon'$. Therefore, if we use $\epsilon'=\tfrac{\epsilon}{2^n}$ then $2^nF'$ is an $\epsilon$-accurate estimate of $\autof(a)^2$.

However, the number of calls to the circuit will be $\Theta(\tfrac{1}{\epsilon'} \log \tfrac{1}{\delta}) = \Theta(\tfrac{2^n}{\epsilon} \log \tfrac{1}{\delta})$ which is $\Omega(2^n)$; this is clearly undesirable and begging to be bettered.

It may be tempting to improve the above method by first amplifying the probability of observing $R_2$ in the state $\ket{0^n}$ and {\em then} estimating the probability of observing $R_3$ in the state $\ket{a}$. However, for amplitude estimation at this stage the probability of $R_2,R_3$ to be in the state $\ket{0^n} \otimes \ket{a}$ should be exactly $c\autof(a)^2$ for some known constant $c$; since $\sigma_f$ is not known, fixed-point amplitude amplification cannot guarantee a knowledge of the exact probability after amplification. Thus it is unclear if amplitude amplification followed by amplitude estimation can lead to a better estimation algorithm.

Now we will describe a quantum algorithm for the aforementioned task aiming for a better query complexity.
Our technical objective will be to
generate a state with a probability that is related to $|\autof(a)|^2$ but much
higher than that in the earlier approach and our main tool will be the quantum technique of ``swap test''.

\begin{figure}[!h]
    \begin{minipage}{0.5\linewidth}
    \centering \leavevmode
	\includegraphics[width=0.6\linewidth]{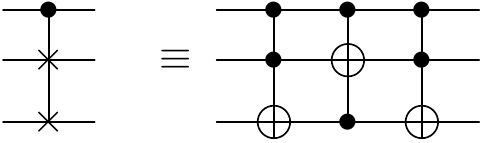}
    \end{minipage}%
    \begin{minipage}{0.5\linewidth}
    \centering \leavevmode
	\includegraphics[width=0.5\linewidth]{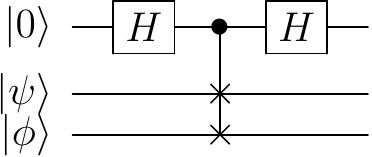}
    \end{minipage}
    \caption{Swap-gate (left) and quantum circuit for swap-test
    (right)\label{fig:swaptest}}
\end{figure}

Suppose we have two registers over the same number of qubits that are in states
denoted by $\ket{\psi}$ and $\ket{\phi}$. The swap test circuit, denoted by $ST$
and illustrated in Figure~\ref{fig:swaptest}, uses an
additional qubit initialized to $\ket{0}$ and applies a
conditional swap-gate in a clever manner such that if the first
(single-qubit) register is measured, then $\ket{0}$ is
observed with probability $\frac{1}{2} [1 + |\braket{\psi}{\phi}|^2]$.
It is easy to show that the circuit performs the following transformation.
\begin{multline*}
\displaystyle \ket{0} \ket{\psi} \ket{\phi} \xrightarrow{ST} \ket{0} \otimes \frac{1}{2}
\Big[\ket{\psi}\ket{\phi} + \ket{\phi}\ket{\psi}\Big] +
\ket{1} \otimes \frac{1}{2} \Big[\ket{\psi}\ket{\phi} -
\ket{\phi}\ket{\psi}\Big]
\end{multline*}

\begin{algorithm}[!h]
    \caption{Autocorrelation estimation at point $a$\label{algo:autocor-est}}
\begin{algorithmic}[1]
    \Require Parameters: $\epsilon$ (confidence), $\delta$ (error)
    \State Start with four registers of which $R_1$ is initialized to $\ket{a}$,
    $R_2$ to $\ket{0}$, and $R_3,R_4$ to $\ket{0^n}$.
    \State Apply these transformations.
	\begin{itemize}
	    \item[] $\ket{a} \ket{0} \ket{0^n} \ket{0^n}$
	    \item[] $\xrightarrow{H^n \otimes H^n} \ket{a} \ket{0} \Big( \sqrtn{n} \sum_x
		\ket{x} \Big) \Big( \sqrtn{n} \sum_y \ket{y} \Big)$
	    \item[] $\xrightarrow{CNOT} \ket{a} \ket{0} \Big( \sqrtn{n} \sum_x
		\ket{x} \Big) \Big( \sqrtn{n} \sum_y \ket{y \xor a} \Big)$
	    \item[] $\xrightarrow{U_f \otimes U_f} \ket{a} \ket{0} \Big( \sqrtn{n} \sum_x
		(-1)^{f(x)} \ket{x} \Big) \Big( \sqrtn{n} \sum_y (-1)^{f(y \xor a)} \ket{y \xor a} \Big)$
		
		\Comment{\small Uses reusable $\ket{-}$}
	    \item[] $\xrightarrow{CNOT} \ket{a} \ket{0} \Big( \sqrtn{n} \sum_x
		(-1)^{f(x)} \ket{x} \Big) \Big( \sqrtn{n} \sum_y (-1)^{f(y \xor a)} \ket{y} \Big)$ 
		
	    \item[] $= \ket{a} \ket{0} \ket{\psi} \ket{\phi_a}$
		\begin{itemize}
		    \item Normalized state $\sqrtn{n} \sum_x (-1)^{f(x)} \ket{x}$ denoted $\psi$
		    \item Normalized state $\sqrtn{n} \sum_y (-1)^{f(y \xor a)} \ket{y}$ denoted $\phi_a$
		\end{itemize}
	\end{itemize}
	\State Apply $ST$ on $R_2,R_3$ and $R_4$ to obtain $$\ket{a} \Big[ \ket{0}
    \otimes \frac{1}{2} \big(\ket{\psi}\ket{\phi_a} + \ket{\phi_a}\ket{\psi}\big) +
    \ket{1} \otimes \frac{1}{2} \big(\ket{\psi}\ket{\phi_a} -
    \ket{\phi_a}\ket{\psi}\big) \Big]$$
	\State $\ell \leftarrow$ estimate the probability of observing $R_2$ in the state $\ket{0}$ 
    with accuracy $\pm\frac{\epsilon}{2}$ and error $\delta$
	\State Return $2\ell-1$ as the estimate of $|\autof(a)|^2$
\end{algorithmic}
\end{algorithm}

\begin{figure}[!h]
	\includegraphics[width=0.7\linewidth]{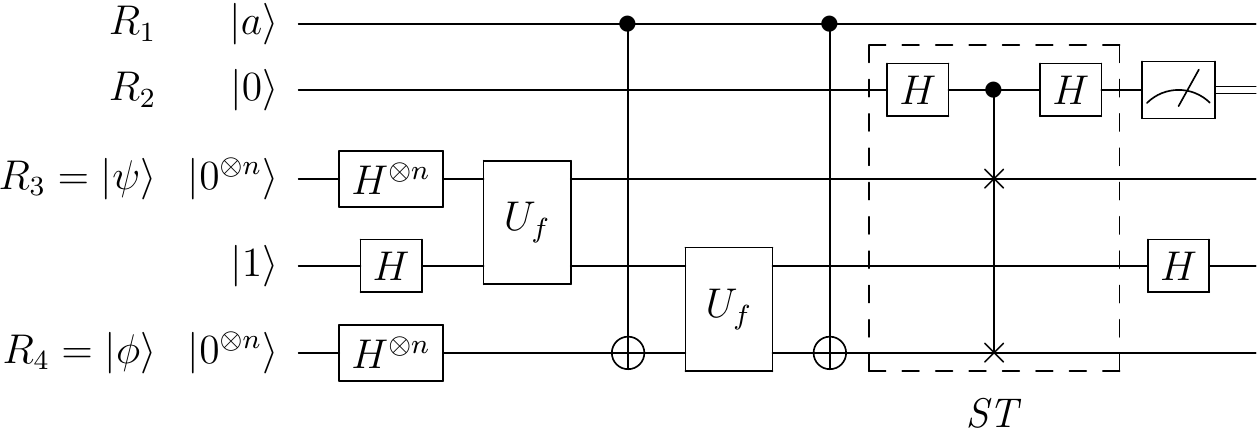}
	\centering \leavevmode
	\caption{Circuit for estimation of autocorrelation coefficient at a
		point $a$\label{fig:estimation}}
\end{figure}

Our algorithm for estimation of $|\autof(a)|^2$ is presented in
Algorithm~\ref{algo:autocor-est} and a circuit diagram is given in
Figure~\ref{fig:estimation}. We do not show the $\ket{1}$ qubit in the algorithm; it is merely used, in the form $\ket{-}$, to apply the $U_f$ gate in a phase-version.

Obviously, an accurate estimation of
$\frac{1}{2} [ 1 + |\autof(a)|^2 ]$ will automatically lead to an accurate
estimation of $|\autof(a)|^2$.
Observe that $\frac{1}{2} [ 1 + |\autof(a)|^2 ] \gg |\autof(a)|^2/2^n$ and
therefore, estimation using Algorithm~\ref{algo:autocor-est} is more efficient
compared to that obtained from autocorrelation sampling (describe earlier in
this section).

\begin{theorem}
    Algorithm~\ref{algo:autocor-est} makes $\Theta\left(\frac{\pi}{\epsilon}\log \frac{1}{\delta}\right)$ calls to $U_f$ and returns an estimate $\alpha$ such that
    $$\Pr \left[ \alpha - \epsilon \le \autof(a)^2 \le \alpha + \epsilon \right] \ge 1-\delta$$
\end{theorem}

\begin{proof}
    Let $\ket{\chi^0_{a}}$ denote the state $\half \ket{\psi}\ket{\phi_a} +
		\half\ket{\phi_a}\ket{\psi}$. Observe that
    $$\left\| \ket{\chi^0_a} \right\|^2 = \tfrac{1}{4} \left[ 2 \big\| \ket{\psi} \big\| \cdot \big\| \ket{\phi_a} \big\| + 2 \big| \braket{\psi}{\phi_a} \big|^2 \right] = \tfrac{1}{2}\left[ 1 + \big| \braket{\psi}{\phi_a} \big|^2\right] $$
    From Step-3 of the algorithm, the probability of observing $R_2$ in the state $\ket{0}$, say denoted $p_0$, can be expressed as $\big\| \ket{\chi_a^0} \big\|^2$.
    
    Further observe that $\braket{\psi}{\phi} = \frac{1}{2^n} \sum_x (-1)^{f(x)} (-1)^{f(x\xor a)} = \autof(a)$. Therefore, $p_0 = \tfrac{1}{2} + \tfrac{1}{2}\autof(a)^2$ and $\ell$ is an estimate of $p_0$ such that,
    \begin{align*}
	1 - \delta & \le \Pr\left[\ell - \tfrac{\epsilon}{2} \le p_0 \le \ell + \tfrac{\epsilon}{2}\right]\\
	& = \Pr[2\ell - \epsilon \le 2p_0 \le 2\ell + \epsilon]\\
	& = \Pr[2\ell - 1 -\epsilon \le 2p_0-1 \le \ell -1 + \epsilon] \\
	& = \Pr[\alpha - \epsilon \le \autof(a)^2 \le \alpha + \epsilon] \tag{$\because$ $\alpha=2\ell-1$}
    \end{align*}
This shows that $2\ell-1$ is an $\epsilon$-accurate estimate of $\autof(a)^2$.
    
    For analysing the number of queries to $U_f$, first observe that the circuit to obtain the state in Step-3 of the algorithm (see Figure~\ref{fig:estimation}) uses only two calls to $U_f$. The amplitude estimation procedure of Lemma~\ref{lemma:ampest} shall make $\Theta\left( \tfrac{\pi}{\epsilon} \log \tfrac{1}{\delta} \right)$ calls to this circuit, giving a total of $\Theta\left( \tfrac{\pi}{\epsilon} \log \tfrac{1}{\delta} \right)$ calls to $U_f$. \qed
\end{proof}

The above theorem shows how to estimate $\autof(a)^2$ using a quantum algorithm that shows a quadratic speedup over a classical sampling-based algorithm. However, there remains the question of estimating $\autof(a)$ when its value is 0. In the above approach, $p_0$ shall be $\tfrac{1}{2}$, and therefore, $\ell \le \tfrac{1}{2} + \epsilon/2$. This implies that the estimate for $\autof(a)^2$ shall only satisfy $\alpha \le \epsilon$. A minor improvement may be added to Algorithm~\ref{algo:autocor-est} to handle this situation that we now describe.

First apply the previously mentioned technique of applying amplitude
estimation on the output state of sampling algorithm from
Section~\ref{sec:sampling} but using a very high $\epsilon$. Note that amplitude estimation does not err when the
probability it is estimating is 0. Then run
Algorithm~\ref{algo:autocor-est} as usual and return the minimum of the two
estimates. In case $\autof(a)=0$, the first amplitude estimation will correctly return 0 as the estimate. We skip the details due to lack of space in this paper.

\subsection{Estimation of Sum-of-Squares Indicator}\label{subsec:est_s_f}
In this section we consider the problem of estimating the sum-of-squares indicator $\sigma_f$.
As before the objective will be to obtain an estimate with $\epsilon$ accuracy and $\delta$ probability of error. Since $\sigma_f \ge 1$, typical values of $\epsilon$ will be 1 or more.

We first discuss a classical sampling-based approach.
Let $a,b,c$ be three random variables chosen uniformly at random from $\Ftwo{n}$
such that $b \not= c$ and let $X_{a,b,c}$ be the $\pm1$-valued random variable $(-1)^{f(a
\xor b)} (-1)^{f(a \xor c)}$.
We first express $\sigma_f$ as the expectation of these random variables.

\begin{align*}
	\sigma_f & = \sum_{a \in \Ftwo{n}} \autof(a)^2 
	= \sum_{a \in \Ftwo{n}} \Big[ \frac{1}{2^n} \sum_{b \in \Ftwo{n}} (-1)^{f(b) \xor f(b \xor a)} \Big]^2 \\
	& = \frac{1}{2^{2n}} \sum_{a \in \Ftwo{n}} \Big[ 2^n + \sum_{{b\not= c
			\atop b,c \in \Ftwo{n}}} (-1)^{f(a \xor b) \xor f(a \xor c)} \Big]\\
	& = 1 + \frac{1}{2^{2n}} \sum_{a \in \Ftwo{n} \atop b\not= c} (-1)^{f(a
		\xor b) \xor f(a \xor c)}\\
	& = 1 + (2^n-1) \E_{a,b,c}[X_{a,b,c}]
\end{align*}

 Note that $\E[X_{a,b,c}] = \frac{\sigma_f - 1}{2^n -
	1} \approx \frac{\sigma_f}{2^n}$. One way to estimate $\E[X_{a,b,c}]$ is to use
multiple independent samples of $a,b,c$. Since each sample of $X_{a,b,c}$
requires 2 calls to $f()$, therefore $O(\frac{1}{\epsilon'^2} \log
\frac{1}{\delta})$ calls to $f()$ would be sufficient to estimate
$\E[X_{a,b,c}]$ with $\epsilon'$ accuracy and $\delta$ error. 
Suppose $\tilde{X}$ is the estimate that we obtain; since it satisfies
$$\Pr[\tilde{X} - \epsilon' \le \E[X_{a,b,c}] \le \tilde{X} + \epsilon'] \ge 1 - \delta$$
then an estimate of $\sigma_f=1+(2^n-1)\E[X_{a,b,c}]$ can be obtained by $1+(2^n-1)\tilde{X}$.
It follows that
$$\Pr\Big[{1+(2^n-1)\tilde{X}} - \epsilon'(2^n-1) \le \sigma_f \le {1+(2^n-1)\tilde{X}} + \epsilon'(2^n-1) \Big] \ge 1-\delta$$
Thus, if we want to estimate $\sigma_f$ with accuracy $\epsilon$, we have to set $\epsilon'=\frac{\epsilon}{2^n-1} \approx \frac{\epsilon}{2^n}$.
The number of calls to $f()$ then becomes $O(\frac{2^{2n}}{\epsilon^2} \log \frac{1}{\delta})$ which is only marginally better than the $\Theta(2^{2n})$ classical non-randomized process of computing all autocorrelation values and then summing them up.

On the quantum side, the circuit in Figure~\ref{fig:autocorr-sampling} can help us in estimating the sum-of-squares indicator of $f$.
Since the probability of observing $R_2$ (in Figure~\ref{fig:autocorr-sampling}) to be in the state $\ket{0^n}$ is $\sigma_f/2^n$, 
Lemma~\ref{lemma:ampest} can be used to efficiently estimate $\sigma_f/2^n$.
The number of calls to $U_f$ shall be $\Theta\left(\frac{2^n}{\epsilon} \log \frac{1}{\delta}\right)$ following the same analysis that was done in Section~\ref{subsec:autocor-estimation}. Thus we get a quadratic improvement over the classical sampling algorithm.

We tried to improve upon this method by using the swap-test technique of Section~\ref{subsec:autocor-estimation} and running Algorithm~\ref{algo:autocor-est} with initial state $\sqrtn{n}\sum_x \ket{x}\ket{0}\ket{0^n}\ket{0^n}$. We can estimate the probability of observing the output qubit in the state $\ket{0}$ using a {\em relative accuracy} quantum estimation approach. However, the number of calls to $U_f$ remained the same $\Theta\left(\frac{2^n}{\epsilon} \log \frac{1}{\delta}\right)$.

\section{Conclusion}
Autocorrelation spectrum is a very important tool for designing Boolean functions with good cryptographic properties and also for mounting
differential attacks of cryptosystems. In this paper we design several efficient quantum algorithms that analyse different aspects of Boolean functions that are related to their autocorrelation spectra. We first show that the Deutsch-Jozsa algorithm can be suitably extended to sample the Walsh spectrum of
any derivative. Further, we specifically concentrate on the autocorrelation spectrum of a Boolean function. We present an algorithm to sample according
to a distribution that is proportional to the autocorrelation coefficients of a Boolean function. Finally we consider the estimation of some values or some functions
of autocorrelation coefficients with high accuracy and small error. Our algorithms will have applications to evaluate the cryptographic properties 
of a Boolean function in a significantly faster manner than in classical paradigm.

\section*{Acknowledgements}
The second author acknowledges the support from the project ``Cryptography \& Cryptanalysis: How far can we bridge the gap between Classical and Quantum Paradigm'', awarded under DAE-SRC, BRNS, India.
\bibliographystyle{plain}

\bibliography{refs}{}

\end{document}